%% file: circuitMonitorsMain.tex
\documentclass[runningheads]{llncs}
\input{styles}
\input{macros}
\author{Luca Aceto \inst{1,3} \and Antonis Achilleos\inst{1} \and Elli Anastasiadi\inst{1}\and Adrian Francalanza\inst{2}}
\institute{ICE-TCS, Department of Computer Science, Reykjavik University, Iceland  \and 
Department of Computer Science, University of Malta, Msida, Malta 
\and Gran Sasso Science Institute, L'Aquila, Italy \\
\email{\href{mailto:luca@ru.is}{luca@ru.is}}, \email{\href{mailto:luca.aceto@gssi.it}{luca.aceto@gssi.it}},
\email{\href{mailto:antonios@ru.is}{antonios@ru.is}},
\email{\href{mailto:elli19@ru.is}{elli19@ru.is}},
\email{\href{mailto:adrian.francalanza@um.edu.mt}{adrian.francalanza@um.edu.mt}},
}

\authorrunning{Aceto et al.}
\titlerunning{Monitoring hyperproperties with circuits}

\title{Monitoring Hyperproperties with Circuits\thanks{The authors were supported by the projects `Open Problems in the Equational Logic of Processes’ (OPEL) (grant No 196050-051) and `Mode(l)s of Verification and Monitorability' (MoVeMent) (grant No~217987) of the Icelandic Research Fund, and `Runtime and Equational Verification of Concurrent Programs' (ReVoCoP) (grant No 222021), of the Reykjavik University Research Fund. Luca Aceto's work was also partially supported by the Italian MIUR  PRIN 2017 project FTXR7S IT MATTERS `Methods and Tools for Trustworthy Smart Systems'.}}
\begin{document}
\maketitle
\begin{abstract}
This paper presents an extension of the safety fragment of Hennessy-Milner Logic with recursion over sets of traces, in the spirit of Hyper-LTL. It then introduces a novel monitoring setup that employs circuit-like structures to combine verdicts from regular monitors. 
The main contribution of this study is the definition of the monitors and their semantics, as well as a monitor-synthesis procedure from formulae in the logic that yields ‘circuit-like monitors’ that are sound and violation complete over a finite set of infinite traces.
\end{abstract}

\section{Introduction}
The field of runtime verification concerns itself with providing methods for checking whether a system satisfies its intended specification at runtime. This runtime analysis is done through a computing device called a \textit{monitor} that observes the current run of a system in the form of a trace \cite{Bartocci2018,francalanza_found_Run_Moni}. Runtime verification has recently been extended to the setting of concurrent systems
 \cite{AcetoPOPL19,Bocchi_Tuosto_DistInteraction,Cassar2017ReliabilityAF,choreogr_monitors_as_mem}
  with several attempts to specify properties over sets of traces,  and to introduce novel monitoring setups \cite{Agrawal2016RuntimeVO,complex_monitor_hyper,monitor_hyper}. A centerpiece in this line of work has been the specification logic Hyper-\ltl \cite{hyperproperties}. Intuitively Hyper-\ltl allows for existential and universal quantification over a set of traces (which describes the set of observed system runs). The properties over one trace are stated in \ltl, with free trace variables, and then made dependent on properties of other traces via the quantification that binds the trace variables. 

We define the linear-time specification logic \hyperhml, as a counterpart to Hyper-\ltl, building on
previous studies of monitorability and monitor synthesis for \muhml \cite{AcetoPOPL19,hml_monitors}, which are necessary for the kind of correctness and complexity guarantees we aim to achieve in this work. 
However, just like Hyper-\ltl, \hyperhml can define dependencies over different traces, which intuitively causes extra delays in the processing of traces as the properties observed on one of them can impact what is expected for another. For example, if a property requires that an event of a trace is compared against an event occurring in all other traces then the processing cost of this event becomes dependent on the number of traces.
 In this approach, we keep the processing-at-runtime cost (as defined in \cite{rabin_real_time_comp}) minimal by restricting the type of properties verified to a natural fragment of \hyperhml, but applying no assumptions on the system under scrutiny. 
  This comes in contrast with the existing research, where the runtime verification of such properties is dealt with via a plethora of modifications and assumptions made over the monitoring setup, such as being able to restart an execution or having access to all executions of a system. 

Our monitor setup is engineered for the studied fragment of the specification language, by utilizing circuit-like structures to combine verdicts over different traces. The fragment of the logic restricts the amount of quantification that can be applied to the properties of individual traces and thus limits the dependencies between them. This naturally induces circuits with monitors from \cite{AcetoPOPL19} as input nodes and simple kinds of gates at the higher levels, with the resulting structure having constant depth with respect to the corresponding formula, which is considered efficient in the field of parallel computation \cite{complexity_small_depth_circuits}. Thus, each step taken by such a monitor in response to an event of the system under scrutiny takes constant time, which makes the monitors `real time' in the sense of \cite{rabin_real_time_comp}.

\section{The logic}
Our logic is defined in the style of Hyper-\ltl as presented in \cite{hyperproperties}. The quantification among traces remains the same, but the language in which local trace properties are stated is \muhml.
%
We consider the following restriction to a multi-trace \shml logic (the \textit{safety} fragment of \muhml \cite{AcetoPOPL19}), with no alternating quantifiers, called \ellihml.  We can similarly define the \chml (co-safety) fragment, and the \hml fragment. 

\begin{definition}
Formulae in $\textsc{Hyper}^1\text{-}\textsc{sHML}\xspace$ are  constructed by the following grammar:
	\begin{align*}
	\varphi \in \textsc{Hyper}^1\text{-}\textsc{sHML}\xspace& ::= \exists_{\pi} \psi &\mid&~ \forall_{\pi} \psi &\mid&~ \varphi \sqcup \varphi &\mid&~ \varphi \sqcap \varphi  
\end{align*}
where
$\psi$ stands for a formula in \shml and $\pi$ is a trace variable from an infinite suppy of trace variables $\mathcal{V}$. 
 $\sqcup \text{ and }\sqcap$ stand for the regular $\vee$ and $\wedge$ boolean connectives, only usable at the top syntax level. 
  Although the syntactic distinction is cosmetic, it allows us to keep the synthesis function in Definition \ref{def:synthesis} clearer.
\end{definition}
\textit{Semantics} The semantics of \hyperhml is given over a finite set of infinite traces $T$ over \act and it is a natural extension of the linear-time semantics of \muhml.  The existential and universal quantification happens via the trace variable $\pi$ which ranges over the traces in $T$. The extension of the \muhml linear-time semantics from \cite{AcetoPOPL19} to the \hyperhml semantics is done in the style of Hyper-\ltl. 
This semantics applies to \ellihml, which is a fragment of \hyperhml.
We only consider \textit{closed} formulae in \ellihml and for these we use the standard notation $T \models \varphi$ to mean that a set of traces $T$ satisfies $\varphi$ (and similarly for $ T \not \models \varphi$).
\begin{example}\label{ex:1} The \ellihml formula  $ \forall_{\pi} [a]\false \sqcap \exists_{\pi} [b] ( max~x. ([a]\false ~\wedge ~[b] x))$,  over the set of actions $\{a,b\}$, states that for any set of traces $T$, none of the traces in $T$ start with $a$, and $b^\omega \in T$.
\end{example} 
\section{The monitors}

The intuition behind our monitor design is the following (we recommend following this intuition along with the example given in Figure \ref{fig:circuit_monitor}).
Over a finite set of traces $T$ we instrument a circuit-like structure.  
Each trace $t \in T$ is assigned a fixed set of regular monitors that correspond to the properties in \shml to be verified.
 These regular monitors are connected with simple gates which evaluate to $yes$, $no$ or $end$ based on the verdicts produced by their associated regular monitors. Once some of these gates start evaluating to verdicts, they communicate with more complex gates, connected in a circuit-like graph, which propagate input verdicts though logic operations until the root node of the circuit reaches a verdict as well. The formal definition of a circuit monitor is given in the style of computational complexity circuits \cite[Definition~1.10]{circuit_complexity}.
\begin{definition}
The language $\text{C}\textsc{mon}_k$ of $k$-ary  monitors, for $k >0$, is given through the following grammar: 
\begin{align*}
&M \in \text{C}\textsc{mon}_k
::= \bigvee [m]_k &&|~~ \bigwedge [m]_k &&| ~~ M \vee M &&|~~ M \wedge M
\\
		& m ::=~~~ yes ~| ~no ~|~end &&|~~ a.m,~ a\in \act &&|~~ m +n &&|~~ rec ~x.m &&|~~ x
\end{align*}
\monc is the collection of infinite \textit{sequences} $(M_i)_{i \in \nat}$ of terms that are generated by substituting $k = i, \forall i \in \nat$, in a term $M
$ in $\text{C}\textsc{mon}_k$. 
\end{definition}


 We use  $M,M' \ldots$ to denote the monitors (infinite sequences of terms generated by the first line of this grammar), and refer to them as circuit monitors, and $m_1, m_2 \ldots$ to denote the regular monitors described by the second line. 
  The notation $[m]_k$ corresponds to the parallel dispatch of $k$ identical regular monitors $m$, where $k = |T|$, with $T =\{t_1, \ldots ,t_k \}$. 

  Given a monitor $M \in \monc$, for some $k >1$, we will call each syntactic sub-monitor of $M$ a  gate. For example, we have inductively that over the monitor $M' \vee M''$ we have the gates $M' \vee M''$ and all gates contained in monitors $M'$, and $M''$, while for the monitor $\bigvee [m]_k$ we have the gates $\bigvee [m]_k$ and gates $m_{[i]}$ for $i \in \{ 1, \ldots, k \}$.
   For $M \in \monc$ we define a set of \textit{program variables} $G_M$, where one variable $g_{M'}$ is assigned to each gate $M'$ of $M$.

   For readability purposes we will be omitting the naming $g$ of the program variables and call them by the name of the gate they represent.
     We use
  $m_{[i]}$ to mean the regular monitor $m$ instrumented over the trace $t_i$. 
  It is important here to see that $g_{m_{[i]}}$ will be the \textit{name} of the gate assigned to one such monitor and stays unchanged while the actual monitor advances its computation as trace events are read.
   This will be clarified later, through the instrumentation rules.  

A program variable related to gate $M$, can be assigned the following values: $yes$, $no$, $end$, and $j$, with $j \in \{ 0, \dots, 2^{(\ell+1)}-1 \}$, $\ell$ being the number of immediate syntactical sub-monitors of gate $M$.
Number $j$ is encoded in binary, and is used to carry the information of which sub-gates have given some verdict (this means that the encoding of $j$ has $\ell +1$ bits).
The value of the $\ell +1$-th bit of $j$ is reserved to encode that one of the sub gates has outputted an $end$. The information that $j$ carries is very important for the evaluation of a gate, as often this evaluation depends on the verdicts of more that one sub-gate, as well as what these verdicts are (see Figure \ref{fig:circuit_monitor}). A variable $g_m$ can only take the values $yes$, $no$ and $end$, produced by the relevant monitor instrumented over a trace.

  A \textbf{\textit{configuration}} of monitor $M \in \monc$, for some $k>1$, is an array $s_M$ containing a value for all program variables $g$ of $M$. We denote the set of all configurations for a monitor $M$ as $\mathcal{S}_M$. We use the notation $s[M\backslash i]$ to denote the update of a configuration $s$ where gate $M$ stores some value $j$ to one where the $i$-th coordinate of $j$ is $0$, while all other variables have the value they had in configuration $s$. Similarly, we use the notation $s[M\backslash end_i]$ to refer to a configuration where the update $s[M\backslash i]$ has taken place \textit{and} the value of the $\ell +1$-th bit of $j$ is set to $1$. We also use the notation $s[\sfrac{v}{M}]$ with $v \in \{yes,no,end\}$,  to mean a configuration where the value of the variable for gate $M$ is updated to $v$.
  
   All gate variables in a circuit monitor are initialized to $2^{\ell}-1$ (a sequence of $\ell$-many zeros), to represent that all sub-gates are waiting to give some output and $s_{M_{init}}$ stands for the initial configuration of $M$. Since $M$ is a family of circuits, we have that the initial configuration of each monitor $M_i$ in the family corresponds to a different initial configuration $s_{M_{i-init}}$.
 \begin{example}
Figure \ref{fig:circuit_monitor}, is an example of a circuit monitor and its evaluation.
\end{example}
\begin{figure}
\begin{minipage}{0.5\textwidth}
\centering
\begin{tikzpicture}[scale=0.6,->, node distance=1.6cm,auto, main node/.style={circle,draw,font=\tiny\bfseries}]
\node (m) [main node]{$\displaystyle\bigwedge_{\{11\}}$};
\node (s) [left of=m] {$s_{M_{3-init}}$};
\node (tr) [right=1.8cm of m] {$\longrightarrow^*$};
\node (c) [main node,below left of=m] {$\displaystyle\bigwedge_{\{111\}}$};
\node (d) [main node,below right of=m] {$\displaystyle\bigvee_{\{111\}}$};
\node (m1) [below left of=c] {$m_{1,1}$};
\node (m2) [right=0.cm of m1] {$m_{2,1}$};
\node (m3) [below right of=c] {$m_{1,2}$};
\node (m4) [right=0.cm of m3] {$m_{2,2}$};
\node (m5) [below right=0.5cm and 0.cm of d] {$m_{1,3}$};
\node (m6) [right=0.cm of m5] {$m_{2,3}$};
\node (t1) [below of=m1] {$a^{\omega}$};
\node (t2) [below left = 1.1cm and 0.cm of m4] {$b.a.b^{\omega}$};
\node (t3) [below of=m6]{$b^{\omega}$};
\path 
	(c) edge (m)
	
    (d) edge (m)
	(m1) edge [dashed] (c)
	(m2) edge [dashed] (d)
	(m3) edge [dashed] (c)
	(m4) edge [dashed] (d)
	(m5) edge [dashed] (c)
	(m6) edge [dashed] (d)
 	(t2) edge [blue,dashed] (m3)
         edge [red,dashed] (m4)
 	(t3) edge [blue,dashed] (m5)
         edge [red,dashed] (m6)       
 	(t1) edge [blue,dashed] (m1)
         edge [red,dashed] (m2);
\end{tikzpicture}
\end{minipage}
\begin{minipage}{0.5\textwidth}
\centering
\begin{tikzpicture}[scale=0.6,->, node distance=1.6cm, auto, main node/.style={circle,draw,font=\tiny\bfseries}]
\node (m) [main node]{\small $\underset{\{01\}}{\displaystyle no}$};
\node (s) [left of=m] {$s_{M}$};
\node (c) [main node,below left of=m] {\small $\underset{\{100\}}{no}$};
\node (d) [main node,below right of=m] {$\displaystyle\bigvee_{\{011\}}$};
\node (m1) [below left of=c] {$no$};
\node (m2) [right=0.cm of m1] {$no$};
\node (m3) [below right of=c] {$yes$};
\node (m4) [right=0.cm of m3] {$m_{2,2}$};
\node (m5) [below right=0.5cm and 0.cm of d] {$yes$};
\node (m6) [right=0.cm of m5] {$m_{2,3}$};
\node (t1) [below of=m1] {$a^{\omega}$};
\node (t2) [below left = 1.1cm and 0.cm of m4] {$a.b^{\omega}$};
\node (t3) [below of=m6]{$b^{\omega}$};
\path 
	(c) edge (m)
	
    (d) edge (m)
	(m1) edge [dashed] (c)
	(m2) edge [dashed] (d)
	(m3) edge [dashed] (c)
	(m4) edge [dashed] (d)
	(m5) edge [dashed] (c)
	(m6) edge [dashed] (d)
 	(t2) edge [blue,dashed] (m3)
         edge [red,dashed] (m4)
 	(t3) edge [blue,dashed] (m5)
         edge [red,dashed] (m6)       
 	(t1) edge [blue,dashed] (m1)
         edge [red,dashed] (m2);
\end{tikzpicture}
\end{minipage}
\caption{The circuit monitor for the formula from Example \ref{ex:1}, over $T = \{a^{\omega}, b.a.b^{\omega}, b^{\omega} \}$.}
\label{fig:circuit_monitor}
\end{figure}
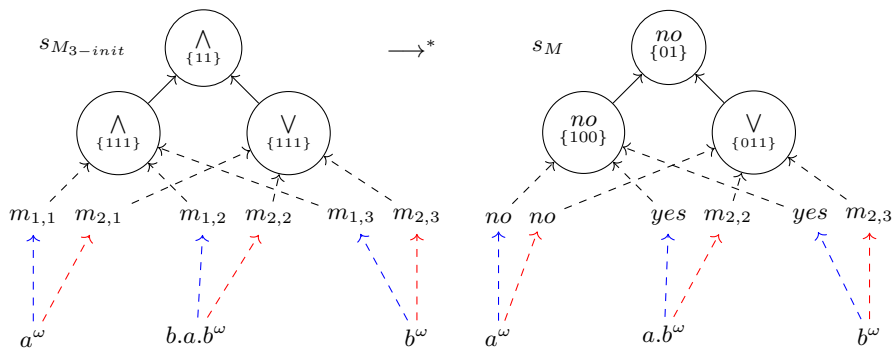
\begin{figure} 
\footnotesize
\textbf{Monitor semantics:}{
\centering
\begin{gather*}
\SOSrule{s[m_{[i]}]=yes}
{s \rightarrow s[\sfrac{yes}{\bigvee [m]_k }]}
\qquad
\SOSrule{s[m_{[i]}]=no}
{s \rightarrow s[\bigvee [m]_k \backslash i]}
\qquad
\SOSrule{s[m_{[i]}]=end}
{s \rightarrow s[\bigvee [m]_k \backslash end_i]}
\\
\SOSrule{s[\bigvee[m]_k]=0}
{s \rightarrow s[\sfrac{no}{\bigvee [m]_k }]}
\qquad
\SOSrule{s[\bigvee[m]_k]=2^k}
{s \rightarrow s[\sfrac{end}{\bigvee [m]_k }]}
\end{gather*}}
\textbf{Instrumentation:}
{\centering
\begin{gather*}
\SOSrule{m \xrightarrow[\text{}]{\tau} m'}
{m \instr t \xrightarrow[\text{}]{\tau} m' \instr t}
\quad
\SOSrule{m \xrightarrow[\text{}]{a} m'}
{m \instr a.t \xrightarrow[\text{}]{a} m' \instr t}
\quad
\SOSrule{ \forall j \in \{ 1,\ldots r\},~ m_{j[i]} \instr t  \xrightarrow[\text{}]{a} m_{j[i]}'\instr t'}
{ s \instr (\overrightarrow{m} \instr T) \rightarrow s \instr (\overrightarrow{m}[\sfrac{m_{j[i]}'}{m_{j[i]},~ \forall j }] \instr T[\sfrac{t'}{t}])}
\\
\SOSrule{s  \xrightarrow[\text{}]{} s'}
{s \instr( \overrightarrow{m} \instr T)  \rightarrow s' \instr ( \overrightarrow{m}\instr T)}
\qquad
\SOSrule{s \instr (\overrightarrow{m}\instr T) \rightarrow s \instr ( \overrightarrow{m}[\sfrac{v}{n_{j[i]}}] \instr T[\sfrac{t'}{t}])}
{s \instr (\overrightarrow{m} \instr T) \rightarrow s[\sfrac{v}{g_{m_{j[i]}}}]\instr (\overrightarrow{m}[\sfrac{v}{n_{j[i]}}] \instr T[\sfrac{t'}{t}])}
\end{gather*}}
\caption{\label{tab:sos_rules} Operational semantics of processes in \monc.}
\end{figure}
\textit{Semantics} 
The semantics of a regular monitors is as presented in \cite{AcetoPOPL19}. Each regular monitor corresponds to an LTS, and a transition labeled with $a \in \act$ corresponds to a regular monitor observing the event $a$ when instrumented with a system $p$ that produces it.
The semantics of a circuit monitor is given as a transition relation $\xrightarrow[\text{}]{} \subseteq$ $\mathcal{S}_M \times \mathcal{S}_M$ and the instrumentation  $\instr$ takes place over a set of regular monitors $\overrightarrow{m}$ instrumented over a set  of traces $T$, denoted $M(T)$.

  We define $M(T): = s_{M_{|T|-init}} \instr \overrightarrow{m}_{[i]} \instr T$, where $\overrightarrow{m}$ is the set of regular monitors that occur in $M$, and $\overrightarrow{m}_{[i]}$ is $ \overrightarrow{m}$, instrumented over the trace $t_i \in T.$ When $m$ is a regular monitor then $\instr$ stands for the existing instrumentation relation from \cite{AcetoPOPL19}.
The transition and instrumentation relations are defined as the least ones that satisfy the axioms and rules in Figure \ref{tab:sos_rules}. Due to lack of space, we only include the rules giving the semantics of the $\bigvee [m]_k$ monitor. Those for the other operators follow the same structure. The proof in Appendix \ref{appendix:proof_viol_comp} could help with the understanding of the more intricate instrumentation rules. 

A monitor is required to be \textit{correct} with respect to some specification formula $\varphi$. The notions of correctness we use in this work are defined below.

\begin{definition} 
Given a monitor $M \in$ \monc , and a set of traces $T$.
\begin{itemize}
\item  $M$ \textbf{rejects} $T$ (resp. \textbf{accepts} $T$) denoted $rej(M,T)$  (resp. $acc(M,T)$) iff $M(T) \rightarrow^* s \instr \overrightarrow{n}\instr T'$ for some $s,\overrightarrow{n},~ T'$, where $s[M]=no$ (resp. $s[M]=yes$). 
\item Given a formula $\varphi \in$ \hyperhml, $M$ is \textbf{sound} for $\varphi$ if $\forall T$, $acc(M,T) \implies T \models \varphi$, and $rej(M,T) \implies T \not \models \varphi$.
\item $M$ is \textbf{violation complete} for $\varphi$ if $\forall T$, $T\not \models \varphi \implies rej(M,T)$.
\end{itemize}
\end{definition} 
\textit{Synthesis:}
Given a formula $\varphi$ in \ellihml, We synthesize a circuit monitor $M$ through the following recursive function $Syn(-) :$ \ellihml $\rightarrow$ \monc. 
\begin{definition}[Circuit Monitor Synthesis]\label{def:synthesis}
\begin{align*} 
&Syn(\exists_{\pi} \varphi)= \bigvee [m(\varphi)]_{k}  &&~~ Syn(\forall_{\pi} \varphi)= \bigwedge [m(\varphi)]_{k}  \\
&Syn( \varphi_1 \sqcup \varphi_2)= Syn(\varphi_1) \vee Syn(\varphi_2)
&&~~ Syn( \varphi_1 \sqcap \varphi_2)= Syn(\varphi_1) \wedge Syn(\varphi_2)
\end{align*}  
 Where $m(-)$ is the monitor synthesis function for \shml defined in \cite{AcetoPOPL19}.
\end{definition}
\begin{proposition}\label{{prop:viol_comp}} Given a formula $\varphi$ in \ellihml, we have $Syn(\varphi)$ is a sound and violation-complete monitor for $\varphi$. 
\end{proposition}
\begin{proof}
 The proof is by induction on the structure of $\varphi$. We present here a characteristic case and give more details for some of them in the Appendix \ref{appendix:proof_viol_comp}.
Assume that $ \varphi = \exists_{\pi} \psi$, with $ \psi \in \shml$ and that we have a set of traces $T$ s.t. $T \not \models \varphi$.
  From the semantics of \ellihml, we have that $t_i \not\models \psi$, for all traces $t_i$ in $T$.
   However $\psi \in \shml$ and thus from \cite{AcetoPOPL19} we get that $m_{\psi}$ is a violation complete monitor for $\psi$. 
This means that for all $t_i \in T$, there exist $t_i' \in \act^*$ and $t_i'' \in \act^{\omega}$, such that $t_i = t_i'.t_i''$, such that the monitor $m_{\psi}$ rejects $t_i'$.

 From the rules in Figure \ref{tab:sos_rules} we see that each gate $g_{m_{\psi[i]}}$ will reach the value $no$ as enough events over the trace $t_i'$ will occur. I.e. $s_M \instr \overrightarrow{m_{\psi_{[i]}}} \instr T \rightarrow^* s_M \instr \overrightarrow{m}_{[i]}[\sfrac{no}{m_{[i]}}] \instr T[\sfrac{t_i''}{t_i}]$, witch propagates to the evaluation of $g_{m_{[i]}}$ to $no$, for all $i$. We now study the transitions $s_M[\sfrac{no}{g_{m_{\psi}[i]}}]$ since those can be then composed with this instrumentation via the fourth instrumentation rule.  Applying the SOS rules yields that the update $\backslash i$ takes place for all $i$ at the gate $\bigvee [m]_k$ which means that the value of $j$ stored in it becomes $0$. This finally yields that the value of the final gate $\bigvee [m]_k$ becomes $0$, i.e. $s_M [\sfrac{no}{g_{m_{[i]}}}~ \forall i]  \rightarrow s_M [\sfrac{no}{\bigvee [m]_k}]$. Since this transition can be composed with the discussed instrumentation we have that $s_M \instr \overrightarrow{m_{\psi_{[i]}}} \instr T \rightarrow s_M[\sfrac{no}{g_{\bigvee [m_{\psi}]_{[i]}}}]\instr \overrightarrow{n} \instr T'$ for some $\overrightarrow{n}$ and $T$ and we are done.  \qed
\end{proof}

\subsection{Runtime costs}

The monitor synthesis in Definition \ref{def:synthesis} provides a family of circuits that can be instrumented appropriately on an arbitrary set of traces to analyze the events occurring in them. Ideally, the runtime cost of monitoring resulting from our constructions should be bounded by a constant that does not depend on the parameters of the system (such as the number of available traces, or of the events observed so far) \cite{rabin_real_time_comp}. In this way, if a monitor is launched along with the system components, it will only induce a feasible computational overhead. 

We already know that the regular monitors instrumented with individual traces analyze the system events they observe with a constant overhead \cite{hml_monitors}. Regarding the computational cost of the circuit part, since we are given $k$ many traces, it must be that the necessary computation performed from a circuit monitor can be  performed in parallel, distributed over the components that produced the traces in the first place. This means that we can only concern ourselves with the \textit{circuit complexity} \cite{circuit_complexity} of a given monitor, which encapsulates the parallel processing power necessary for its evaluation.

We now observe the synthesis function. There, a formula $\varphi$ in \ellihml will be turned into a family of circuit monitors where, for each connective of the original formula $\varphi$, the output monitor increases in size based on the size for the monitors of the sub-formulae of $\varphi$. However, for each connective of the formula, the \textit{depth} of the circuit is only increased by $1$ which means that the output circuit monitor has a depth bounded by the size of the formula $\varphi$. Since the gates of the output monitor can have either a fixed amount of sub-gates ($\vee, \wedge$), or $k$ many ($\bigvee, \bigwedge$), we have that the output circuit is in the complexity class AC$^0$ \cite{circuit_complexity}. Thus, the monitor only adds a constant computational overhead when executed over the computational resources of the distributed components of the system.

\section{Conclusion and future work}

We expect that the fragment \ellihml is maximal with respect to violation completeness, which means that any monitor in \monc is monitoring for a formula in \ellihml. 
 However, the ultimate goal of this work is to extend the collection of monitorable properties by allowing alternating quantifiers in the syntax.
  This is a very important aspect of any work in this field, as the more interesting hyperproperties, such as the property ``at all times, if one trace encounters the event $p$ then all traces do so as well'' which is a necessary component for the expression of properties such as noninference \cite{temp_logic_hyperprop,possibilistic_prop}, require alternation of quantifiers. 
  
  A way to tackle this would be to project such properties into the \ellihml fragment. However this procedure is not formally yet defined, or trivial and one could argue that since every hyperproperty has been shown (\cite{hyperproperties}) to be the intersection of a liveness and a safety hyperproperty, (and since liveness and safety properties are widely accepted as independent \cite{alpernBowenSchneiderLiveandSafe}), an elimination of alternating quantifiers can only take place in very few cases.
    Thus, our main purpose is to extend the logic and the consequent monitors in order to express and monitor for the most general class of such properties.
     The main objective of the logical fragment we give here is to establish a formal baseline which we will attempt to extend in future work.

 Our approach to an extension would be to allow a notion of synchronization rounds among the regular monitors (or equivalently a round of communication). This would enable more complex dependencies between traces, as now the properties required of a given trace can be impacted by the state of the ones monitored for on a different one. However, the analysis of communications among the monitors is a complicated extension, as their exact content plays a significant role to our insight over the system, as well as the processing at runtime cost. We plan to implement this therefore by utilizing dynamic epistemic logic \cite{dyn_epist_logic} in order to perform this extension formally and soundly. 


\bibliographystyle{plain}
\bibliography{thebibliography.bib}
\newpage
\appendix
\section{Appendix: cases for the proof of violation completeness}\label{appendix:proof_viol_comp}
Here we give some more insight on the remaining cases of the violation completeness proof. First we highlight that the second base case of our proof, for formulae of the form $\forall_{\pi} \psi$ is completely analogous to the one we give and thus omitted.

 We will here give an important lemma necessary for analyzing both remaining cases, and then present the high level details for the case of $\sqcap$. The intuition of the importance of the lemma is that the monitors $Syn(\varphi_1)$ and $Syn(\varphi_2)$ should not have their computation affected from the fact that they are run in parallel over a set of traces $T$. 
 
\begin{lemma}
If 
\begin{itemize}
\item $s_{M_1} \instr \overrightarrow{m_1}[i] \instr T\rightarrow s_{M_1}'\instr \overrightarrow{m_1}[i]' \instr T'$, and 
\item $s_{M_2} \instr \overrightarrow{m_2}[i] \instr T\rightarrow s_{M_2}'\instr \overrightarrow{m_2}[i]' \instr T'$
\end{itemize} 
then 
\begin{itemize}
\item $s_{M_1 \vee M_2}  \instr \overrightarrow{m_{12}}[i] \instr T\rightarrow s_{M_1 \wedge M_2}' \instr \overrightarrow{m_{12}}[i]' \instr T'$, and 
\item $s_{M_1 \wedge M_2}  \instr \overrightarrow{m_{12}}[i] \instr T\rightarrow s_{M_1 \wedge M_2}' \instr \overrightarrow{m_{12}}[i]' \instr T'$,  
\end{itemize}
where $\overrightarrow{m_{12}} = \overrightarrow{m_{2}} \cup \overrightarrow{m_{2}}$ and $\overrightarrow{m_{12}}' = \overrightarrow{m_{2}}' \cup \overrightarrow{m_{2}}'$ respectively.
\end{lemma}

\begin{proof}
We note here that a configuration for $s_{M_1 \vee M_2}$ is identical to one for $s_{M_1 \wedge M_2}$ except the root variable, as all other variables they both contain are $s_{M_1}' \cup s_{M_2}'$. 

The key aspect of this proof is the third rule of the instrumentation relation. There we can see that in order for a configuration instrumented over a set of regular monitors, instrumented over a set of traces, can only advance its computation, if all monitors instrumented over the same trace progress with their computation synchronously by reading the next trace event. 

Thus, form the assumptions of this lemma we get that for all $j = \{1,\ldots r\}$, where $r$ is the total amount of different regular monitors occurring in $M_1$ and $M_2$ the premise of our rule is satisfied and thus the cumulative configuration of variables amounting for the union of variables of the two circuit monitors $M_1$ and $M_2$ (including the root variable), can perform the necessary transition to the new state, where all regular monitors (those both from $M_1$ and $M_2$) assigned to trace $t_i$ have processed the event $a$, and we are done. \qed
\end{proof}

Having the above lemma streamlines our inductive step for the rest of the cases. 
Assuming a non-base-case formula in \ellihml we can clearly see that it must be of the form $\varphi = \varphi_1 \sqcap \varphi_2$   or $\varphi = \varphi_1 \sqcap \varphi_2$. We only analyze one of the two cases as they are symmetrical. 
For any set of traces $T$, such that $T \not\models \varphi$, from the semantics of \ellihml, we have that $T \not\models \varphi_1$ and $T \not\models \varphi_2$. Since the synthesized monitor for $\varphi_1 \sqcap \varphi_2$ can reach a configuration where the values of the gates for $Syn(\varphi_1)$ and $Syn(\varphi_2)$ are the same as they would be for the individual monitors instrumented over $T$, and by inductive hypothesis (which guarantees that $Syn(\varphi_1)$ and $Syn(\varphi_2)$ are violation-complete) we have necessary conclusion by combining the two negative verdicts of the individual monitors via the semantics. \qed

\end{document}

%% file: styles.tex
\usepackage[utf8]{inputenc}
\usepackage{amsmath}
\usepackage{amsfonts}
\usepackage{amssymb}
\usepackage{stmaryrd}
\usepackage{hyperref}
\usepackage{bussproofs}
\usepackage{hyperref}
\usepackage{xfrac}
\usepackage{enumerate}
\usepackage[shortlabels]{enumitem}
\usepackage[]{algorithm2e}
\usepackage{pgf}
\usepackage{tikz}
\usepackage{authblk}
\usetikzlibrary{arrows.meta}
\usetikzlibrary{shapes}
\usetikzlibrary{graphs}
\usetikzlibrary{matrix,chains,positioning,decorations}
 \usepackage{makeidx}
\tikzset{
  treenode/.style = {shape=rectangle, rounded corners,
                     draw, align=center,
                     top color=white, bottom color=blue!20},
  root/.style     = {treenode, font=\Large, bottom color=red!30},
  env/.style      = {treenode, font=\ttfamily\normalsize},
  dummy/.style    = {circle,draw}
}

%% file: macros.tex
\newcommand*{\hml}{\textsc{HML}\xspace}
\newcommand*{\muhml}{$\mu$\textsc{HML}\xspace}

\newcommand*{\shml}{\textsc{sHML}\xspace}
\newcommand*{\chml}{\textsc{cHML}\xspace}

\newcommand*{\hyperhml}{Hyper-$\mu$\textsc{HML}\xspace}
\newcommand*{\ellihml}{\textsc{Hyper}$^1$-\textsc{sHML}\xspace}

\newcommand*{\ltl}{LTL\xspace}

\newcommand*{\act}{\ensuremath{\textsc{Act}}\xspace}

\newcommand{\instr}{\triangleleft}

\newcommand{\false}{\mathtt{ff}}

\newcommand{\SOSrule}[2]{\frac{\displaystyle #1}{\displaystyle #2}}

\newcommand*{\nat}{\ensuremath{\mathbb{N}}}

\newcommand{\monc}{\text{C}\textsc{mon}\xspace}

%% file: circuitMonitorsMain.bbl
\begin{thebibliography}{10}

\bibitem{AcetoPOPL19}
Luca Aceto, Antonis Achilleos, Adrian Francalanza, Anna
  Ing{\'{o}}lfsd{\'{o}}ttir, and Karoliina Lehtinen.
\newblock Adventures in monitorability: from branching to linear time and back
  again.
\newblock {\em Proc. {ACM} Program. Lang. POPL}, 3(52):1--29, 2019.

\bibitem{Agrawal2016RuntimeVO}
Shreya Agrawal and Borzoo Bonakdarpour.
\newblock Runtime verification of k-safety hyperproperties in {HyperLTL}.
\newblock In {\em {IEEE} 29th Computer Security Foundations Symposium, {CSF}
  2016, Lisbon, Portugal, June 27 - July 1, 2016}, pages 239--252. {IEEE}
  Computer Society, 2016.

\bibitem{alpernBowenSchneiderLiveandSafe}
Bowen Alpern and Fred~B. Schneider.
\newblock Recognizing safety and liveness.
\newblock {\em Distrib. Comput.}, 2(3):117–126, sep 1987.

\bibitem{Bartocci2018}
Ezio Bartocci, Yli{\`{e}}s Falcone, Adrian Francalanza, and Giles Reger.
\newblock Introduction to runtime verification.
\newblock In Ezio Bartocci and Yli{\`{e}}s Falcone, editors, {\em Lectures on
  Runtime Verification - Introductory and Advanced Topics}, volume 10457 of
  {\em Lecture Notes in Computer Science}, pages 1--33. Springer, 2018.

\bibitem{Bocchi_Tuosto_DistInteraction}
Laura Bocchi, Kohei Honda, Emilio Tuosto, and Nobuko Yoshida.
\newblock A theory of design-by-contract for distributed multiparty
  interactions.
\newblock In Paul Gastin and Fran{\c{c}}ois Laroussinie, editors, {\em CONCUR
  2010 - Concurrency Theory}, pages 162--176, Berlin, Heidelberg, 2010.
  Springer Berlin Heidelberg.

\bibitem{complex_monitor_hyper}
Borzoo Bonakdarpour and Bernd Finkbeiner.
\newblock The complexity of monitoring hyperproperties.
\newblock In {\em 31st {IEEE} Computer Security Foundations Symposium, {CSF}
  2018, Oxford, United Kingdom, July 9-12, 2018}, pages 162--174. {IEEE}
  Computer Society, 2018.

\bibitem{Cassar2017ReliabilityAF}
Ian Cassar, Adrian Francalanza, Claudio~Antares Mezzina, and Emilio Tuosto.
\newblock Reliability and fault-tolerance by choreographic design.
\newblock In Adrian Francalanza and Gordon~J. Pace, editors, {\em Proceedings
  Second International Workshop on Pre- and Post-Deployment Verification
  Techniques, PrePost@iFM 2017, Torino, Italy, 19 September 2017}, volume 254
  of {\em {EPTCS}}, pages 69--80, 2017.

\bibitem{temp_logic_hyperprop}
Michael~R. Clarkson, Bernd Finkbeiner, Masoud Koleini, Kristopher~K. Micinski,
  Markus~N. Rabe, and C{\'e}sar S{\'a}nchez.
\newblock Temporal logics for hyperproperties.
\newblock In Mart{\'i}n Abadi and Steve Kremer, editors, {\em Principles of
  Security and Trust}, pages 265--284, Berlin, Heidelberg, 2014. Springer
  Berlin Heidelberg.

\bibitem{hyperproperties}
Michael~R. Clarkson and Fred~B. Schneider.
\newblock Hyperproperties.
\newblock {\em J. Comput. Secur.}, 18(6):1157--1210, 2010.

\bibitem{dyn_epist_logic}
Hans~van Ditmarsch, Wiebe van~der Hoek, and Barteld Kooi.
\newblock {\em Dynamic Epistemic Logic}.
\newblock Springer Publishing Company, Incorporated, 1st edition, 2007.

\bibitem{monitor_hyper}
Bernd Finkbeiner, Christopher Hahn, Marvin Stenger, and Leander Tentrup.
\newblock Monitoring hyperproperties.
\newblock {\em Formal Methods Syst. Des.}, 54(3):336--363, 2019.

\bibitem{francalanza_found_Run_Moni}
Adrian Francalanza, Luca Aceto, Antonis Achilleos, Duncan~Paul Attard, Ian
  Cassar, Dario~Della Monica, and Anna Ing{\'{o}}lfsd{\'{o}}ttir.
\newblock A foundation for runtime monitoring.
\newblock In Shuvendu~K. Lahiri and Giles Reger, editors, {\em Runtime
  Verification - 17th International Conference, {RV} 2017, Seattle, WA, USA,
  September 13-16, 2017, Proceedings}, volume 10548 of {\em Lecture Notes in
  Computer Science}, pages 8--29. Springer, 2017.

\bibitem{hml_monitors}
Adrian Francalanza, Luca Aceto, and Anna Ing{\'{o}}lfsd{\'{o}}ttir.
\newblock Monitorability for the hennessy-milner logic with recursion.
\newblock {\em Formal Methods Syst. Des.}, 51(1):87--116, 2017.

\bibitem{complexity_small_depth_circuits}
Johan H\r{a}stad.
\newblock {\em Computational Limitations of Small-Depth Circuits}, volume~53.
\newblock MIT Press, Cambridge, MA, USA, 1987.

\bibitem{possibilistic_prop}
John McLean.
\newblock A general theory of composition for a class of ``possibilistic''
  properties.
\newblock {\em {IEEE} Trans. Software Eng.}, 22(1):53--67, 1996.

\bibitem{choreogr_monitors_as_mem}
Claudio~Antares Mezzina and Jorge~A. P\'{e}rez.
\newblock Causally consistent reversible choreographies: A monitors-as-memories
  approach.
\newblock In {\em Proceedings of the 19th International Symposium on Principles
  and Practice of Declarative Programming}, PPDP '17, page 127–138, New York,
  NY, USA, 2017. Association for Computing Machinery.

\bibitem{rabin_real_time_comp}
Michael~O. Rabin.
\newblock Real time computation.
\newblock {\em Israel Journal of Mathematics}, 1(4):203--211, 1963.

\bibitem{circuit_complexity}
Heribert Vollmer.
\newblock {\em Introduction to Circuit Complexity - {A} Uniform Approach}.
\newblock Texts in Theoretical Computer Science. An {EATCS} Series. Springer,
  1999.

\end{thebibliography}
